\documentclass[journal]{IEEEtran}

\usepackage{amsmath}
\usepackage{booktabs}
\usepackage{xcolor}
\usepackage{amssymb}
\usepackage{amsthm}
\usepackage{fontawesome5}
\usepackage{graphicx} 
\usepackage{multirow} 
\usepackage{mathtools} 
\usepackage{comment}
\usepackage{algpseudocode}
\usepackage{algorithm}
\usepackage{booktabs}
\usepackage{siunitx}
\newtheorem{theorem}{Theorem}

\usepackage{tikz}
\usetikzlibrary{arrows.meta, positioning}

\usepackage{eso-pic}

\AddToShipoutPictureFG{%
  \AtPageLowerLeft{%
    \raisebox{.7 cm}{%
      \makebox[\paperwidth]{%
        \normalfont\rmfamily\footnotesize
        Approved for Public Release on 22 Apr 2026. Distribution Is Unlimited; Case Number: AFRL-2026-2017%
      }%
    }%
  }%
}

\ifCLASSINFOpdf
\else
\fi

\begin{document}
\title{Adversarial Robustness of Near-Field Millimeter-Wave Imaging under Waveform-Domain Attacks}


\author{Lhamo~Dorje,~Jordan~Madden,~Soamar~Homsi,~and~Xiaohua~Li
\thanks{L. Dorje, J. Madden, and X. Li are with the Department of Electrical and Computer Engineering, Binghamton University, Binghamton, NY 13902 USA (e-mail: ldorje1@binghamton.edu; jmadden2@binghamton.edu; xli@binghamton.edu).}%

\thanks{S. Homsi is with the Air Force Research Laboratory, Information Directorate, Rome, NY 13441 USA (e-mail: soamar.homsi@us.af.mil).}

\thanks{Manuscript received March xx, 2026; revised September 00, 00.}}


\maketitle
\begin{abstract}
Near-field millimeter-wave (mmWave) imaging is widely deployed in safety-critical applications such as airport passenger screening, yet its own security remains largely unexplored. This paper presents a systematic study of the adversarial robustness of mmWave imaging algorithms under waveform-domain physical attacks that directly manipulate the image reconstruction process. 
We propose a practical white-box adversarial model and develop a differential imaging attack framework that leverages the differentiable imaging pipeline to optimize attack waveforms. We also construct a real measured dataset of clean and attack waveforms using a mmWave imaging testbed. 
Experiments on 10 representative imaging algorithms show that mmWave imaging is highly vulnerable to such attacks, enabling an adversary to conceal or alter targets with moderate transmission power. Surprisingly, deep-learning-based imaging algorithms demonstrate higher robustness than classical algorithms. These findings expose critical security risks and motivate the development of robust and secure mmWave imaging systems.

\end{abstract}

\begin{IEEEkeywords}
Millimeter-wave Imaging, Adversarial Attack, Robustness, Adversarial Machine Learning, Synthetic Aperture Radar, RF Sensing
\end{IEEEkeywords}

\IEEEpeerreviewmaketitle
\section{Introduction}
\label{sec:introduction}
Near-field millimeter-wave (mmWave) imaging enables high-resolution two-dimensional (2D) or three-dimensional imaging capabilities comparable to optical systems while offering unique advantages such as all-weather operation, penetration through clothing, non-ionizing radiation, and privacy preserving. These properties have led to widespread deployment in safety- and security-critical applications, including airport passenger screening, concealed weapon detection, public surveillance, and senior care \cite{wang2019review, van2021millimeter}. 

Despite its growing adoption, the security of the mmWave imaging pipeline itself, especially the image reconstruction algorithms, remains largely unexplored. This gap is particularly critical for high-resolution near-field mmWave imaging systems. Although existing research on radar or radio frequency (RF) sensing security has extensively studied signal-level attacks such as jamming, interference, and spoofing
\cite{han2025rf}, it primarily focuses on far-field point-target tracking rather than high-resolution imaging. In parallel, adversarial machine learning has extensively studied perturbations in the digital image domain \cite{mei2024comprehensive, chen2025analyzing}. However, it primarily focuses on the adversarial robustness of downstream imaging recognition models. It does not address whether an adversary can manipulate physical sensing waveforms to alter the reconstructed images. 

This paper addresses this critical gap by studying the adversarial robustness of mmWave imaging algorithms when adversaries inject physical signal waveforms to cause the reconstruction of fake but visually plausible mmWave images. These adversarial attacks aim to compromise the image reconstruction, thereby evading the detection of both human and automated analysis. Our central premise is that the mmWave imaging pipeline can be implemented as a differentiable mapping from received waveforms to reconstructed images, which enables the use of gradient-based optimization to design adversarial waveforms that achieve specific attack objectives. 

As a preliminary study, our work in \cite{dorje2024evasive} demonstrated that simple waveform injection can produce camouflaged imaging results. This stealthy attack raises immediate security concerns; for example, an adversary carrying a weapon could potentially pass airport screening by deploying a carefully designed attack device near the imaging system. However, it remains unclear whether such vulnerabilities are inherent to mmWave imaging in general or limited to specific algorithms. This motivates a systematic evaluation across a diverse set of imaging algorithms.

The contributions of this paper are summarized as follows:
\begin{enumerate}
\item To the best of our knowledge, this paper presents the first systematic evaluation of adversarial robustness for high-resolution near-field mmWave imaging algorithms at the image reconstruction stage under waveform-domain physical attacks.



\item This paper establishes a new adversarial study framework in adversarial machine learning, contributing a strong yet realistic white-box threat model, an end-to-end differentiable imaging pipeline, and a differential imaging attack (DIA) algorithm for optimizing physical attack waveforms.


\item A new dataset consisting of real measured imaging and attack waveforms is constructed. The adversarial robustness of 10 representative mmWave imaging algorithms is evaluated under various attack strategies. The results demonstrate that mmWave imaging is not robust to adversarial attacks and, contrary to common belief, learning-based algorithms appear more robust than classical algorithms.   
\end{enumerate}



The remainder of this paper is organized as follows. Section \ref{sec:literaturereview} reviews related work. Section \ref{sec:operators} presents the mmWave imaging model and representative reconstruction algorithms. Section \ref{sec:threat_model} introduces the adversarial model and the proposed attack framework. Section \ref{sec:experiments} describes the experimental setup and presents the evaluation results. Section \ref{sec:conclusion} concludes the paper. 
Source code and dataset are publicly available at \footnote{https://github.com/ldorje1/Differential-Imaging-Attacks-on-Near-Field-SAR-Imaging}. 

\section{Literature Review} \label{sec:literaturereview}

RF sensing and imaging \cite{kong2024survey} play important roles across a wide range of applications, including medical imaging \cite{chao2012millimeter}, senior care \cite{gurbuz2019radar}, and security and surveillance \cite{chetty2011through}. Remote sensing using satellites or unmanned aerial vehicles (UAVs) has been widely adopted for ecological monitoring \cite{dahhani2022land}, crisis management \cite{ge2020review}, and public surveillance \cite{robin2016review}. Most of these systems operate in the far field and provide relatively low-resolution imaging, suitable for tasks such as point-target localization, human activity recognition, gesture recognition, and large-scale geological mapping. 

In contrast, near-field mmWave imaging, enabled by short wavelengths and short sensing distances, can achieve high-resolution imaging comparable to optical cameras \cite{yanik2019near}. It has been widely used in applications such as passenger screening in airports, train stations, and government facilities \cite{wang2019review}, \cite{bian2024towards}, as well as concealed weapon detection \cite{sheen2001three, garcia2018combining}. High-resolution imaging typically requires a large sensing aperture. Fortunately, low-cost implementations can employ a single-antenna sensor that is mechanically scanned over a large aperture to collect measurements, followed by image reconstruction based on the synthetic aperture radar (SAR) principle \cite{bian2024towards, li2021lightweight}. A variety of mmWave imaging algorithms have been developed, ranging from classical analytical methods such as the range migration algorithm to modern deep neural network-based approaches \cite{wang2019review}. 

Adversarial attacks on RF sensing and imaging systems have been extensively studied \cite{han2025rf}. However, most existing works focus on disrupting signal reception through jamming, interference, or spoofing \cite{psiaki2016gnss, morong2019study}. Furthermore, they primarily consider far-field, low-resolution radar sensing or mmWave systems \cite{kim2024systematic}, \cite{yan2016can, manesh2019cyber, sun2021control, nallabolu2022emulation, ordean2022millimeter, hunt2023madradar, xie2024universal}. It has even been suggested in \cite{vennam2023mmspoof} that high-resolution mmWave imaging may mitigate such attacks; in contrast, this paper demonstrates the opposite. 

As a subfield of adversarial machine learning, numerous studies have examined the robustness related to radar and remote sensing images \cite{ruan2024survey_sar_adv, akhtar2018threat}. These studies are typically conducted in the digital pixel domain and aim to disrupt downstream target recognition or classification models, rather than the image reconstruction process itself. A clear gap remains in understanding the adversarial robustness of near-field, high-resolution mmWave imaging systems under attacks conducted in the physical signal domain to directly manipulate the imaging algorithms.


\section{Near-Field High-Resolution mmWave Imaging}
\label{sec:operators}

\begin{figure}[tbp]
\centerline{
\includegraphics[width=0.55\textwidth, trim=3pt 3pt 3pt 3pt, clip]{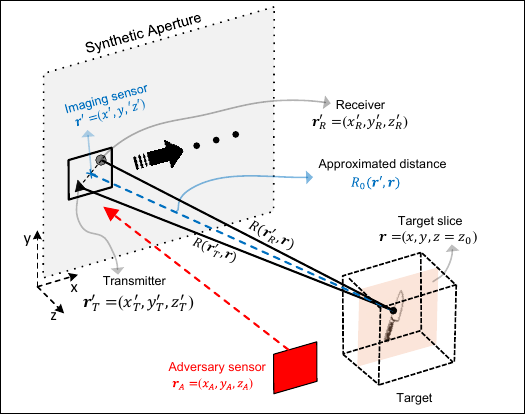}
}
\caption{System model of the near-field mmWave imaging scenario under adversarial attack. The \emph{imaging sensor} performs synthetic aperture scanning to reconstruct an image of the \emph{target}, while a nearby \emph{adversarial sensor} injects attack waveforms into the receiver to induce a fake image.}
\label{fig:sar_sat_img_2}
\end{figure}

\subsection{Imaging Model}
\label{subsec:forward}

As illustrated in Fig.~\ref{fig:sar_sat_img_2}, we consider a near-field mmWave imaging system equipped with an imaging sensor. The sensor sequentially scans the scene and employs the SAR principle to reconstruct the target image. An adversary is assumed to be present in proximity to the scene but remains unknown to the imaging system. During the sensing process, the adversary injects carefully designed waveforms into the receiver, thereby perturbing the measured signals and altering the reconstructed image.



The imaging sensor is assumed to have one transmit antenna and one receive antenna, and it moves on a planar aperture at $z=0$. At each sensing location $\mathbf{r}'=(x', y', z'=0)$, referred to as an \emph{aperture look} $\ell$, the sensor transmits a sensing waveform and records the corresponding echo signal. The objective of the imaging algorithm is to estimate the 2D reflectivity distribution of a target located on the plane $\mathbf{r}=(x, y, z=z_0)$ from the collection of received echoes.

In near-field imaging, the separation between the transmit and receive antennas cannot be neglected. At aperture look $\ell$, the sensor transmits a frequency-modulated continuous-wave (FMCW) signal 
$p(t)$ from position $\mathbf{r}'_{\mathrm{T}}=(x'_{\mathrm{T}}, y'_{\mathrm{T}}, z'_{\mathrm{T}}=0)$, and the echo is recorded at $\mathbf{r}'_{\mathrm{R}}=(x'_{\mathrm{R}}, y'_{\mathrm{R}}, z'_{\mathrm{R}}=0)$. The transmitted waveform is given by
\begin{equation}
p(t)=\exp\!\left\{j2\pi\!\left(f_0 t+\tfrac{1}{2}K t^2\right)\right\}, 
\label{eq:clean_tx_wave}
\end{equation}
where $f_0$ denotes the start frequency and $K$ is the chirp slope.

The received echo signal at aperture look $\ell$ can be modeled as
\begin{equation}
{s}_{\ell}(t)
=
\iiint \sigma(\mathbf{r})\,p\!\big(t-\tau_{\ell}(\mathbf{r})\big)\,d\mathbf{r}
\;+\;
\nu_{\ell}(t),
\label{eq:rf_mix_new}
\end{equation}
where $\sigma(\mathbf{r})$ represents the reflectivity of the target at location $\mathbf{r}$, $\tau_{\ell}(\mathbf{r})$ is the propagation delay from the transmit antenna to $\mathbf{r}$ and then to the receive antenna, and $\nu_{\ell}(t)$ denotes noise, clutter, and interference.

To facilitate algorithmic processing, the imaging region is discretized into $N$ voxels centered at $\{\mathbf r_n\}_{n=0}^{N-1}$, and measurements are collected over $L$ aperture locations $\{\mathbf r'_\ell\}_{\ell=0}^{L-1}$. From \eqref{eq:rf_mix_new}, the standard linear imaging model can be derived \cite{yanik2019near, zeng2025finufft_mmwave} (See Appendix \ref{appendix:yHxmodel} for deduction):
\begin{equation}
\mathbf y = \mathbf H\,\boldsymbol{\alpha} + \mathbf v, 
\label{y=Hxmodel}
\end{equation}
where $\mathbf{y} \in\mathbb C^{L}$ is the complex measurement vector, $\boldsymbol{\alpha} \in\mathbb C^{N}$ denotes the discretized reflectivity, $\mathbf v \in \mathbb C^{L}$ represents noise, and $\mathbf H\in\mathbb C^{L\times N}$ is the propagation matrix with entries
\begin{equation}
H_{\ell n}=
\frac{\exp\!\big(j\,2k\,R_0(\mathbf r'_\ell,\mathbf r_n)\big)}
     {R_0^{2}(\mathbf r'_\ell,\mathbf r_n)},  
\label{eq:H_matrix}
\end{equation}
where $R_0(\mathbf r'_\ell, \mathbf r_n)$ denotes the distance between the sensing location $\mathbf r'_\ell$ and the voxel location $\mathbf r_n$, and $k=2\pi f_0/c$ with $c$ being the speed of light. 

\subsection{Imaging Algorithms}  \label{subsec:imgalg}



The complete imaging pipeline can be expressed as a compositional operator
\begin{equation}
   \mathcal{G}(\cdot) =(\mathcal{G}_{\mathrm{img}} \circ \mathcal{G}_{\mathrm{pre}})(\cdot),
   \label{eq:differentiable_operator}
\end{equation}
where $\mathcal{G}_{\mathrm{pre}}(\cdot)$ represents signal pre-processing (i.e., extracting $\mathbf{y}$ from received waveforms $s_\ell(t)$), and $\mathcal{G}_{\mathrm{img}}(\cdot)$ performs image reconstruction which is essentially estimating $\boldsymbol{\alpha}$ from $\mathbf y$ given $\mathbf H$.

In this paper, we categorize image reconstruction algorithms into two classes: \emph{classical}, and \emph{learning-based}.

The first class consists of classical (non-deep-learning) approaches. One important category is the analytical reconstruction algorithms with closed-form expressions. A canonical example is the time-domain Back-Projection Algorithm ({\bf BPA}), which reconstructs the image via
\begin{equation}
\hat{\boldsymbol{\alpha}} = \mathcal{G}_{\mathrm{img}}(\mathbf{y}) = \mathbf{H}^{\mathsf{H}}\mathbf{y},
\label{eq:physics_bpa}
\end{equation}
where $(\cdot)^\mathsf{H}$ denotes the Hermitian transpose \cite{ccetin2014sparsity}.

Due to the high computational cost of explicitly forming $\mathbf{H}$, frequency-domain implementations are often employed when both sensing and imaging grids are regular. This leads to the Range-Migration Algorithm ({\bf RMA}), which leverages fast Fourier transform (FFT) operations for efficient computation \cite{sheen2001three}. 
For near-field scenarios, the Matched-Filtering Algorithm ({\bf MFA}) has been specifically designed \cite{yanik2019near, yanik2020millimeter}. For irregular sensing geometries, the Light-weight Imaging Algorithm ({\bf LIA}) reconstructs the image iteratively with reduced complexity \cite{li2021lightweight}.

Another important category in this first class formulates image reconstruction as an inverse problem with regularization. A representative example is the Compressive-Sensing-based Algorithm ({\bf CSA}) \cite{ccetin2014sparsity, wei2014sparse}, which solves
\begin{equation}
\hat{\boldsymbol{\alpha}} = \arg\min_{\boldsymbol{\alpha} \in \mathbb{C}^{N}} 
\frac{1}{2}\big\|\mathbf{y} - \mathbf{H}\boldsymbol{\alpha}\big\|_2^2 
+ \lambda\,\mathcal{R}(\boldsymbol{\alpha}),
\label{eq:cs_generic}
\end{equation}
where $\mathcal{R}(\cdot)$ encodes structural priors (e.g., sparsity), and $\lambda$ controls the regularization strength.

The second class consists of learning-based approaches inspired by deep learning. One important category is algorithm unrolling, where iterative solvers for \eqref{eq:cs_generic} are converted into finite-depth computational graphs. A representative example is the joint Range-Migration and Iterative Shrinkage Threshold algorithm ({\bf RMIST}) \cite{wang2021rmist}, which iteratively updates
\begin{equation}
\boldsymbol{\alpha}^{(k+1)} = \mathrm{Prox}_{\theta_k}\!\Big( 
\boldsymbol{\alpha}^{(k)} 
- \mu_k\,\mathbf{H}^{\mathsf{H}}\big(\mathbf{H}\boldsymbol{\alpha}^{(k)} - \mathbf{y}\big) 
\Big),
\label{eq:unrolled_generic}
\end{equation}
where $\mu_k$ is the step size and $\mathrm{Prox}_{\theta_k}$ denotes a learned proximal operator. Another important category is the fully data-driven approaches that employ deep neural networks (DNNs) to directly map measurements or intermediate reconstructions to final images. Examples include Vision Transformer-based models ({\bf ViT}) \cite{smith2022vision} and U-Net-based architectures such as {\bf Deep2S}, {\bf Deep2SP+}, and {\bf CV-Deep2S} \cite{manisali2024efficient}. 

\section{Waveform-Domain Adversarial Attack}
\label{sec:threat_model}

\subsection{Adversarial Model}
\label{subsec:objective_capabilities}

To evaluate adversarial robustness under a worst-case scenario, we adopt a strong yet practical {\bf white-box adversary model} characterized as follows:

\begin{itemize}
\item \emph{Adversary knowledge:} The adversary is assumed to have full knowledge of the imaging system, including sensing locations, sensing timing, and transmitted waveforms. The adversary can also observe or reconstruct the signal waveforms received by the imaging sensor. With this information, the adversary can design and transmit attack waveforms that are synchronized in both time and frequency with the sensing process.

\item \emph{Adversary constraints:} The adversary has no direct access to or control over the internal components of the mmWave imaging system. In addition, the adversary is subject to power constraints, either in terms of per-look transmit power or total transmit power across all sensing locations.
\end{itemize}

It is reasonable to assume that the adversary possesses substantial knowledge of the imaging system. 
For example, technical specifications of commercial mmWave imaging systems are often publicly available through regulatory filings. Many systems rely on unsecured communication protocols, making it feasible for an adversary to intercept, record, and analyze sensing waveforms. 

The assumption that the adversary knows the imaging sensor's received signal waveform $s_\ell(t)$ in \eqref{eq:rf_mix_new} is strong but realistic in many practical scenarios. For example, the adversary may install a passenger screening system identical to the systems deployed in airports. This allows it to replicate and observe the received echo signals $s_{\ell}(t)$. Nevertheless, this strong assumption is not necessary if attack objective can be reduced, which we will demonstrate in experiments.

\subsection{Attack Signal Waveform}
\label{subsec:rf_injection}

As illustrated in Fig.~\ref{fig:sar_sat_img_2}, the adversary injects an attack waveform $p_\ell^{\mathrm{atk}}(t)$ directly into the receiver of the imaging sensor. The resulting received signal becomes the superposition of the clean echo \eqref{eq:rf_mix_new}, the injected attack signal, and noise:
\begin{equation}
{s}_{\ell}(t)
=
\iiint \underbrace{\sigma(\mathbf{r})\,p\!\big(t-\tau_{\ell}(\mathbf{r})\big)\,d\mathbf{r}}_{\text{clean echo}}
\;+\;
\underbrace{p^{\mathrm{atk}}_{\ell}(t)}_{\text{attack signal}}
\;+\; \nu_{\ell}(t).
\label{eq:rf_mix_new2}
\end{equation}

After standard signal processing, the corresponding discrete measurement model becomes
\begin{equation}
\mathbf y = \mathbf H\,\boldsymbol{\alpha} + \mathbf{y}^{\mathrm{atk}} + \mathbf{v},
\label{eq:x_vic_linear}
\end{equation}
where $\mathbf{y}^{\mathrm{atk}}$ denotes the contribution of the injected attack signal.

Since the imaging system is unaware of the attack, it reconstructs the image as
\begin{equation}
\hat{\boldsymbol{\alpha}}_{\mathrm{adv}} \;=\; \mathcal{G}_{\mathrm{img}}(\mathbf y).
\label{eq:full_chain}
\end{equation}
As a result, the reconstructed image $\hat{\boldsymbol{\alpha}}_{\mathrm{adv}}$ deviates from the clean reconstruction $\hat{\boldsymbol{\alpha}}$ due to the presence of $\mathbf{y}^{\mathrm{atk}}$.

\begin{theorem} \label{theorem1}
Consider the aperture look $\ell$ when the imaging sensor is located at phase-center $\mathbf r'_{\ell}$, the target is located at $\mathbf{r}$ and the attack sensor is located at ${\mathbf r}_A$. If the attack signal waveform is
\begin{equation}
p^{\mathrm{atk}}_\ell(t) = w_\ell\, p(t)\,e^{-j2\pi\delta_\ell t},
\label{eq:atk_tx_new}
\end{equation}
with the optimizable complex gain $w_{\ell}\in\mathbb{C}$ and fixed frequency offset 
\begin{equation}
 \delta_\ell = K(\tau_\ell(\mathbf{r}) - \tau_\ell^{\mathrm{atk}}(\mathbf{r}_A)),
\end{equation}
where $\tau_\ell(\mathbf{r})$ and $\tau_\ell^{\mathrm{atk}}(\mathbf{r}_A)$ are the propagation delays of the clean sensing signal and the attack signal, respectively, then the attack signal can be successfully received and processed by the imaging sensor, and the $\ell$th element of $\mathbf{y}^{\mathrm{atk}}$ in \eqref{eq:x_vic_linear} is given by
\begin{equation}
   y_\ell^{\mathrm{atk}} = w_\ell e^{j2kR(\mathbf{r}_\ell', \mathbf{r}_A)}.    
\label{eq:atk_sig_after_proc}
\end{equation}
\end{theorem}

\begin{proof}
See Appendix \ref{app:delay_equiv}.
\end{proof}


\subsection{Differential Imaging Attack (DIA)}
\label{sec:differentiable_operator}

The objective of the attack is to optimize the complex weights $w_\ell$ such that the resulting {\bf adversarial image} $\hat{\boldsymbol{\alpha}}_{\mathrm{adv}}$ appears visually plausible but differs from the {\bf true image} $\hat{\boldsymbol{\alpha}}$. 

The optimal attack weights $w_\ell$ can in principle be derived analytically based on \eqref{eq:x_vic_linear} and \eqref{eq:atk_sig_after_proc}. However, in order to avoid calculating the big matrix $\mathbf{H}$, a more general and flexible approach is to treat the entire imaging pipeline \eqref{eq:differentiable_operator} as a differentiable mapping and perform gradient-based optimization.

Define the attack parameter vector
\begin{equation}
    {\bf w} = \left[\begin{array}{ccc} w_0, & \cdots, & w_{L-1} \end{array} \right].
\end{equation}
The adversary solves the following constrained optimization problem:
\begin{equation}
\mathbf w^\star
=
\arg\min_{\mathbf w}
\;
\left\|
\hat{\boldsymbol{\alpha}}_{\mathrm{adv}}(\mathbf w)-
\hat{\boldsymbol{\alpha}}_{\mathrm{tgt}}
\right\|_2^2
+\lambda\,\|\mathbf w\|_2^2,
\label{eq:attack_opt_argmin}
\end{equation}
where $\hat{\boldsymbol{\alpha}}_{\mathrm{adv}}(\mathbf w)$ denotes the reconstructed image under attack weights $\mathbf w$, $\hat{\boldsymbol{\alpha}}_{\mathrm{tgt}}$ is the adversary-chosen {\bf target image}, and $\lambda$ controls the attack power.

Under the white-box model, the adversary can synthesize the received signal \eqref{eq:rf_mix_new2} and calculate $\hat{\boldsymbol{\alpha}}_{\mathrm{adv}}(\mathbf w)$ for any candidate $\mathbf w$. The optimization problem \eqref{eq:attack_opt_argmin} can therefore be solved using gradient descent, analogous to targeted adversarial attacks in machine learning.

The regularization term $\lambda \|\mathbf{w}\|_2^2$ enforces a constraint on the attack power. Let $P_s = \int | p(t)|^2 dt$ denote the sensing signal power and $P_a = \int | p_\ell^{\mathrm{atk}}(t) |^2 dt$ denote the attack power. Then,
\begin{equation}
    \| \mathbf{w} \|_2^2 = \frac{P_a}{P_s}.   
\label{eq:power_ratio}
\end{equation}
Additional constraints, such as per-look power limits $|w_\ell|^2$, can also be incorporated.

When the attacker uses \eqref{eq:attack_opt_argmin} to optimize the attack signal \eqref{eq:atk_tx_new}, the method is referred to as the Differential Imaging Attack ({\bf DIA}). Alternatively, the attacker may choose $\mathbf{w}$ randomly without optimization. In that case, the attack may only produce random (which is in fact blank) adversarial images rather than a targeted adversarial image $\hat{\boldsymbol{\alpha}}_{\mathrm{tgt}}$. In both cases, the adversary may further refine $\mathbf{w}$ online using intercepted signals during the sensing process.

The overall attack procedure is summarized in Algorithm~1.

\begin{algorithm}[t]
\caption{Waveform-Domain Adversarial Attack}
\begin{algorithmic}[1]
\State \textbf{Before Attack}:

\State \hspace*{1em}{\it Attacker}: Acquire imaging system parameters. Either select a target image $\hat{\boldsymbol{\alpha}}_{\mathrm{tgt}}$ and use the DIA algorithm to pre-compute the optimal $\mathbf{w}^\star$ \eqref{eq:attack_opt_argmin}, or initialize $\mathbf{w}^\star$ randomly.

\State \textbf{During Attack}: At each aperture look $\ell$,
 
\State \hspace*{1em} \emph{Imaging Sensor}: Transmits sensing signal $p(t)$ \eqref{eq:clean_tx_wave}
\State \hspace*{1em} \emph{Attacker}: Intercepts the sensing signal, analyzes it, refines $\mathbf{w}^\star$ using DIA, and transmits the attack signal $p_\ell^{\mathrm{atk}}(t)$ \eqref{eq:atk_tx_new}

\State \hspace*{1em} \emph{Imaging Sensor}:
Receives the echo signal $s_{\ell}(t)$ \eqref{eq:rf_mix_new2}

\State \textbf{Image Reconstruction}:
 \State \hspace*{1em}{\it Imaging Sensor}: Processes $s_\ell(t)$ to obtain $\mathbf{y}$ and reconstructs the image $\hat{\boldsymbol{\alpha}}_{\mathrm{adv}}$ \eqref{eq:x_vic_linear}-\eqref{eq:full_chain}

\label{pseudo_algo}
\end{algorithmic}
\end{algorithm}

\section{Experiments and Simulations} \label{sec:experiments}

We adopt a \emph{measurement-driven simulation} framework to evaluate the adversarial robustness of mmWave imaging algorithms. Specifically, we first construct a dataset consisting of real-measured clean sensing waveforms and attack waveforms. We then use this dataset to simulate imaging pipelines and adversarial attacks, enabling systematic and reproducible evaluation across multiple algorithms and attack strategies.

Such a hybrid approach is necessary because fully hardware-based evaluation is currently not available. Implementing waveform-domain attacks requires stringent time and frequency synchronization between the imaging sensor and the attacker, which is not supported by the low-cost mmWave hardware available today \cite{miura2019low}. Real demonstration of FMCW radar attacks has been reported in only several papers \cite{miura2019low, sun2021control, nallabolu2022emulation, ordean2022millimeter}, and all of them use custom-built sensors and are limited to simple point-target location rather than imaging. Moreover, high-resolution SAR imaging requires raster scanning over a large aperture, which can take hours per image, making large-scale hardware experiments prohibitively time-consuming.

\subsection{Real-Measured Dataset Acquisition}
\label{sec:hardware_acquisition}

As shown in Fig.~\ref{fig:test_bed}, we build a mmWave imaging and attack testbed using two Texas Instruments TI IWR1843Boost radar sensors. One operates as the imaging sensor, while the other acts as the adversarial transmitter.

The imaging sensor is mounted on a two-axis translation stage controlled by a host PC for configuration and data acquisition. The radar operates in the $76$--$81$\,GHz band. The stage performs raster scanning over a planar synthetic aperture of size $250\,\mathrm{mm} \times 250\,\mathrm{mm}$ with step sizes $\Delta x = 1$\,mm and $\Delta y = 2$\,mm, resulting in $250 \times 125 = 31{,}250$ aperture looks. At each look, $256$ fast-time samples are recorded at a sampling rate of $5$\,MHz. The target is positioned at approximately $z_0 \approx 230$\,mm from the sensor. The adversarial sensor is placed near the target and transmits signals toward the imaging receiver.

To obtain an attack signal waveform, two back-to-back measurement runs are conducted. In the first run, only the imaging sensor transmits, producing a clean received waveform. In the second run, the imaging sensor transmits again while the attack sensor injects a time-aligned attack waveform, resulting in a received signal that contains both the clean echo and the attack component. The attack waveform is then extracted from these two measurements. Let the received signals from the two runs be $s_1(t)$ and $s_2(t)$, respectively. We estimate the optimal scaling factor $a$ and time shift $d$ by solving
\begin{equation}
    \{a^\star, d^\star\} = \arg\min_{\{a, d\}} E[ \| s_2(t)-as_1(t-d) \|^2], 
\end{equation}
where $E[\cdot]$ denotes expectation.
The attack waveform is then obtained as $s_2(t)-a^\star s_1(t-d^\star)$.

In practice, extracting high-quality attack waveforms is significantly more challenging than collecting clean data due to synchronization requirements. Many trials are often needed to obtain a usable attack signal waveform.

Using this setup, we collected clean imaging data for more than 10 distinct targets (more than $312{,}500$ waveforms) and approximately $40{,}000$ attack waveforms.

Example waveforms are shown in Fig.~\ref{fig:test_bed_waveform}. In the frequency domain, the clean signal exhibits a single dominant peak corresponding to the true target at 0.01 cycles/sample. In contrast, the attacked signal shows an additional peak at 0.075 cycles/sample, which corresponds to a \emph{false target} introduced by the adversarial waveform. The extracted attack waveform isolates this artificial component, confirming that the injected signal can emulate a physically consistent target response.

\begin{figure}[tbp]
\centerline{
\includegraphics[width=0.48\textwidth, trim=1pt 1pt 1pt 1pt, clip]{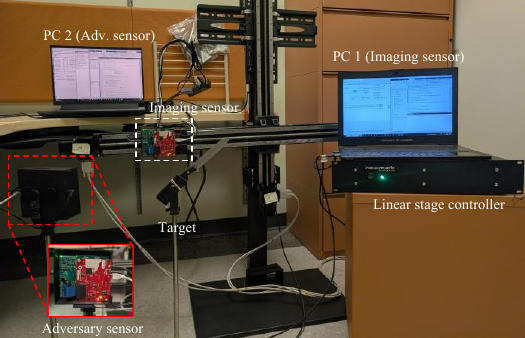}
}
\caption{Experimental testbed for data acquisition. A TI IWR1843Boost mmWave radar mounted on a two-dimensional translation stage performs synthetic aperture scanning of the target, while a second radar acts as an adversarial transmitter.}
\label{fig:test_bed}
\end{figure}

\begin{figure}[tbp]
\centerline{
\includegraphics[width=0.48\textwidth, trim=5pt 1pt 5pt 8pt, clip]{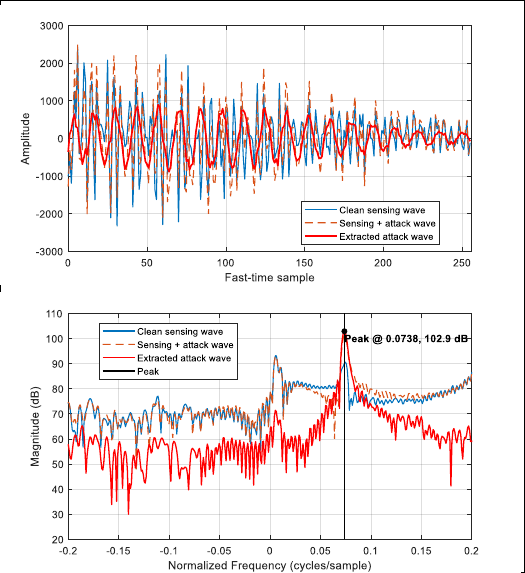}
}
\caption{Example of measured waveforms and their spectra. Top: time-domain signals including the clean echo, the combined sensing-plus-attack waveform, and the extracted attack waveform. Bottom: corresponding normalized magnitude spectra. The attack introduces an additional spectral peak, which corresponds to a physically plausible but false target.}
\label{fig:test_bed_waveform}
\end{figure}

\subsection{Adversarial Attack Simulation Setup}

We evaluate adversarial robustness by simulating imaging and attack pipelines (Algorithm 1) using the real-measured dataset.

We implement 10 representative imaging algorithms:
Five classical algorithms ({\bf BPA}, {\bf LIA}, {\bf RMA}, {\bf MFA}, {\bf CSA}) and five learning-based algorithms ({\bf RMIST}, {\bf ViT}, {\bf Deep2S}, {\bf Deep2SP+}, {\bf CV-Deep2S}). See Section~\ref{subsec:imgalg} for details.
For each algorithm, we implement both the standard (non-differentiable) version used by the imaging system, and a differentiable version used by the attacker for gradient-based optimization. 

For learning-based algorithms, we use publicly available implementations when possible. We made extensive efforts to look for learning-based methods with publicly available code or pretrained models. RMIST and ViT provide source code but no pretrained weights, so we retrained them using our simulated mmWave dataset \cite{dorje2026millisarimagenet}. The only work providing both code and pretrained models is \cite{manisali2024efficient}, from which we adopt Deep2S, Deep2SP+, and CV-Deep2S. However, these three models were originally trained and thus could only work on simulated data rather than real measured mmWave data.

For attacking, we evaluate three attack strategies:
\begin{itemize}
\item {\bf Target Conceal Attack}: The attacker aims to suppress the true target. A blank image is used as $\hat{\boldsymbol{\alpha}}_{\mathrm{tgt}}$, and DIA \eqref{eq:attack_opt_argmin} is used to compute the optimal $\mathbf{w}^\star$;

\item {\bf Target Swap Attack}: The attacker replaces the true target with a fake one. A randomly selected image from the dataset is used as $\hat{\boldsymbol{\alpha}}_{\mathrm{tgt}}$, and DIA is used to compute $\mathbf{w}^\star$;

\item {\bf Randomization Attack}: The attacker uses randomly selected weights $\mathbf{w}$ without DIA optimization.
\end{itemize}

We quantify adversarial robustness or attack effectiveness using Peak Signal-to-Noise Ratio (PSNR), Structural Similarity Index Measure (SSIM), and attack power ratio $P_a/P_s$ in \eqref{eq:power_ratio}. Specifically, we report
\begin{itemize}
\item PSNR$_{AC}$, SSIM$_{AC}$: adversarial vs. clean image
\item PSNR$_{AT}$, SSIM$_{AT}$: adversarial vs. target image
\end{itemize}
Successful attacks reduce PSNR$_{AC}$ and SSIM$_{AC}$ while increasing PSNR$_{AT}$ and SSIM$_{AT}$.
Adversarial robustness implies the opposite behavior. 

For each imaging algorithm, we evaluate performance across all 10 true target objects in the dataset. For each object, the three attack strategies are applied repeatedly using 10 different adversarial target images. Consequently, each performance metric is computed from 100 measurements, from which both the mean and variance are reported.



\begin{table*}[t]
\centering
\caption{Summary of adversarial robustness results. The table reports average PSNR, SSIM, and attack power ratio $P_a/P_s$, averaged over 5 classical and 5 learning-based algorithms, respectively. Global Average is the average of each column.}

\label{tab:summary_results}
\small
\setlength{\tabcolsep}{5pt}
\renewcommand{\arraystretch}{1.08}
\begin{tabular}{c|c||cc|cc|c}
\toprule
Imaging & Attack & \multicolumn{2}{c|}{Adv vs Clean Image} &  \multicolumn{2}{c|}{Adv vs Target Image} & Power  \\
Alg.  &  Strategy & PSNR$_{AC}$(dB) & SSIM$_{AC}$ & PSNR$_{AT}$(dB) & SSIM$_{AT}$ & $P_a/P_s$ \\
\midrule
\footnotesize
   &  Conceal & 18  & .13  & 40 & .85 & .23 \\
Classical & Swap & 10  & .28  & 54 & .95 & .38  \\
   & Random & 3  & .01 &  -  & -  & 10   \\
\midrule 
    & Conceal & 14  & .18 & 24 & .57  & 4.6 \\
Learning & Swap & 18 & .50 & 30 & .83  & 1.5 \\
    & Random & 14  & .21  & - & - & 10  \\
\midrule \midrule
\multicolumn{2}{c||} {Global Average} & 13 & .22 & {\bf 37} & {\bf .80}  & - \\

\bottomrule

\end{tabular}
\end{table*}

\subsection{Results}

\paragraph{Overall performance}
Due to the large volume of data, we first present the summary Table \ref{tab:summary_results} that summarizes the average PSNR and SSIM across all algorithms. Taking a look at the global averages, it is clear that adversarial images $\hat{\boldsymbol{\alpha}}_{\mathrm{adv}}$ are highly similar to the attacker-specified target images $\hat{\boldsymbol{\alpha}}_{\mathrm{tgt}}$ (PSNR$_{AT}=37$\,dB, SSIM$_{AT}=0.80$) while being significantly different from the clean images $\hat{\boldsymbol{\alpha}}$ (PSNR$_{AC}=13$\,dB, SSIM$_{AC}=0.22$). Importantly, the attacks require only moderate power. In many cases, the attack power is comparable to or only slightly higher than the sensing signal power, indicating that the attacks are successful and the imaging algorithms are \emph{not robust} to adversarial attacks.

Another key observation is that learning-based imaging algorithms exhibit \emph{higher robustness} than classical algorithms. In both Target Conceal and Target Swap attacks, higher attack power (4.6 and 1.5) is required by learning-based algorithms than classical algorithms (0.23 and 0.38). Despite this higher power, the resulting PSNR$_{AT}$ and SSIM$_{AT}$ values are lower, suggesting reduced attack effectiveness.


\paragraph{Target Conceal Attack} Detailed results for the target conceal attack are presented in Tables \ref{tab:target_conceal_nopower} and \ref{tab:target_conceal_power}. The former corresponds to unconstrained attack power (used to calculate the ``Conceal'' averages in the summary table), while the latter enforces a constraint of $P_a/P_s=0.1$. Example images are shown in Fig.~\ref{fig:conceal_atk}. 

Results in Table \ref{tab:target_conceal_nopower} show that the attacks effectively conceal the targets. The adversarial images deviate significantly from the clean images and closely resemble the target images, as also illustrated in Fig.~\ref{fig:conceal_atk}. 

Table \ref{tab:target_conceal_power} shows that attacks remain effective against classical algorithms even with limited power (10\,dB below sensing power). However, for learning-based algorithms, the attack becomes less effective under the same power constraint, as PSNR$_{AT}$ and SSIM$_{AT}$ are lower than PSNR$_{AC}$ and SSIM$_{AC}$. This indicates that the adversarial images remain closer to the clean images than to the target images. Overall, learning-based algorithms exhibit greater robustness, requiring higher attack power or resulting in lower attack quality.

\paragraph{Target Swap Attack} Results are shown in Tables \ref{tab:target_swap_nopower} and \ref{tab:target_swap_power}, corresponding to unconstrained and constrained ($P_a/P_s=0.1$) settings, respectively. The data in Table \ref{tab:target_swap_nopower} are used to calculate ``Swap'' averages in the summary Table. Example images are shown in Fig.~\ref{fig:swap_atk}. The results confirm that the attack successfully replaces the true target with the adversarial target. Learning-based methods again demonstrate higher robustness compared to classical methods. 

As a special note, PSNR and SSIM in this case are computed over a region of interest (ROI) containing the object, rather than the entire image. This is because the background structure and dynamic range are similar across images, making full-image metrics less sensitive to differences that are visually apparent.

\paragraph{Randomization Attack with Random Weights $\mathbf{w}$}
Results are presented in Table \ref{tab:target_conceal_random}, with example images in Fig.~\ref{fig:random_attack_success}. With moderately increased attack power ($10$\,dB), the attack successfully leads to blank reconstructed images that differ significantly from the true images.

\begin{table*}[t]
\centering
\caption{Performance of the Target Conceal Attack without power constraints (mean $\pm$ std). The adversarial images achieve low similarity to clean images and high similarity to target images, indicating effective concealment.}
\label{tab:target_conceal_nopower}
\small
\setlength{\tabcolsep}{5pt}
\renewcommand{\arraystretch}{1.08}
\begin{tabular}{l|ccccc}
\toprule
Alg. & PSNR$_{AC}$ & SSIM$_{AC}$ & PSNR$_{AT}$ & SSIM$_{AT}$ & $P_a/P_s$ \\
\midrule
\footnotesize
BPA       & $17.20 \pm 3.07$ & $0.0664 \pm 0.0561$ 
 & $83.68 \pm 24.29$ & $1.0000 \pm 0.0000$ & $0.0614 \pm 0.0690$ \\
LIA       & $17.99 \pm 3.07$ & $0.1156 \pm 0.0971$ & $22.51 \pm 2.83$ & $0.8896 \pm 0.0770$ & $0.0611 \pm 0.0686$ \\
RMA       & $20.98 \pm 3.40$ & $0.2615 \pm 0.1213$ & $28.30 \pm 2.26$ & $0.8726 \pm 0.0689$ & $0.4464 \pm 0.3014$ \\
MFA       & $18.47 \pm 3.30$ & $0.1201 \pm 0.0797$ & $23.99 \pm 0.95$ & $0.5027 \pm 0.0689$ & $0.4809 \pm 0.3014$ \\
CSA       & $18.73 \pm 2.85$ & $0.0790 \pm 0.0550$ & $42.84 \pm 8.47$ & $0.9985 \pm 0.0017$ & $0.0793 \pm 0.1062$ \\ \hline

RMIST & $16.49 \pm 1.41$ & $0.1770 \pm 0.1133$ & $27.50 \pm 3.81$ & $0.7881 \pm 0.1235$ & $12.4791 \pm 2.7903$ \\

ViT & $13.71 \pm 1.49$ & $0.0481 \pm 0.0152$ & $25.99 \pm 2.15$ & $0.7771 \pm 0.0243$ & $1.8767 \pm 0.6811$ \\
Deep2S & $13.47 \pm 1.97$ & $0.1980 \pm 0.1458$ & $23.05 \pm 5.71$ & $0.5088 \pm 0.2439$ & $2.2098 \pm 1.7449$ \\
Deep2SP+ & $15.05 \pm 1.45$ & $0.2863 \pm 0.1117$ & $22.44 \pm 3.98$ & $0.4476 \pm 0.2004$ & $1.5752 \pm 0.9503$ \\
CV-Deep2S & $12.85 \pm 2.59$ & $0.2063 \pm 0.1437$ & $18.73 \pm 5.52$ & $0.3441 \pm 0.2073$ & $4.8752 \pm 3.5686$ \\
\bottomrule
\end{tabular}
\end{table*}

\begin{table*}[t]
\centering
\caption{Performance of the Target Conceal Attack under a power constraint ($P_a/P_s = 0.1$). Classical imaging algorithms remain highly vulnerable, whereas learning-based methods exhibit partial robustness, as reflected by reduced similarity to the target image.}
\label{tab:target_conceal_power}
\small
\setlength{\tabcolsep}{5pt}
\renewcommand{\arraystretch}{1.08}
\begin{tabular}{l|ccccc}
\toprule
Alg. & PSNR$_{AC}$ & SSIM$_{AC}$ & PSNR$_{AT}$ & SSIM$_{AT}$ & $P_a/P_s$ \\
\midrule
\footnotesize
BPA & $17.65 \pm 2.78$ & $0.0833 \pm 0.0566$ & $76.59 \pm 36.25$ & $0.9739 \pm 0.0572$ & $0.10$ \\
LIA & $18.49 \pm 2.56$ & $0.1354 \pm 0.0930$ & $21.77 \pm 3.25$ & $0.8591 \pm 0.1113$ & $0.10$ \\
RMA & $22.11 \pm 3.17$ & $0.4040 \pm 0.0641$ & $18.78 \pm 5.95$ & $0.4009 \pm 0.2972$ & $0.10$ \\
MFA & $20.72 \pm 3.64$ & $0.2115 \pm 0.0987$ & $19.18 \pm 5.17$ & $0.5264 \pm 0.3351$ & $0.10$ \\  
CSA & $19.45 \pm 2.81$ & $0.1044 \pm 0.0649$ & $38.04 \pm 13.42$ & $0.9430 \pm 0.1246$ & $0.10$ \\  \hline
RMIST & $33.98 \pm 8.78$ & $0.5965 \pm 0.2731$ & $18.20 \pm 2.06$ & $0.2805 \pm 0.1160$ & $0.10$ \\
ViT & $21.38 \pm 2.56$ & $0.5026 \pm 0.1728$ & $18.70 \pm 1.33$ & $0.3255 \pm 0.0745$ & $0.10$ \\
Deep2S & $18.10 \pm 1.17$ & $0.5433 \pm 0.0831$ & $11.70 \pm 1.72$ & $0.0897 \pm 0.0789$ & $0.10$ \\
Deep2SP+ & $18.03 \pm 1.04$ & $0.5427 \pm 0.0684$ & $9.75 \pm 4.38$ & $0.1057 \pm 0.0709$ & $0.10$ \\
CV-Deep2S & $21.53 \pm 2.58$ & $0.7726 \pm 0.0970$ & $9.76 \pm 4.03$ & $0.0545 \pm 0.0657$ & $0.10$ \\
\bottomrule
\end{tabular}
\end{table*}

\begin{figure*}[t]
\centering
\includegraphics[width=\textwidth,
trim=1pt 1pt 1pt 1pt,clip]{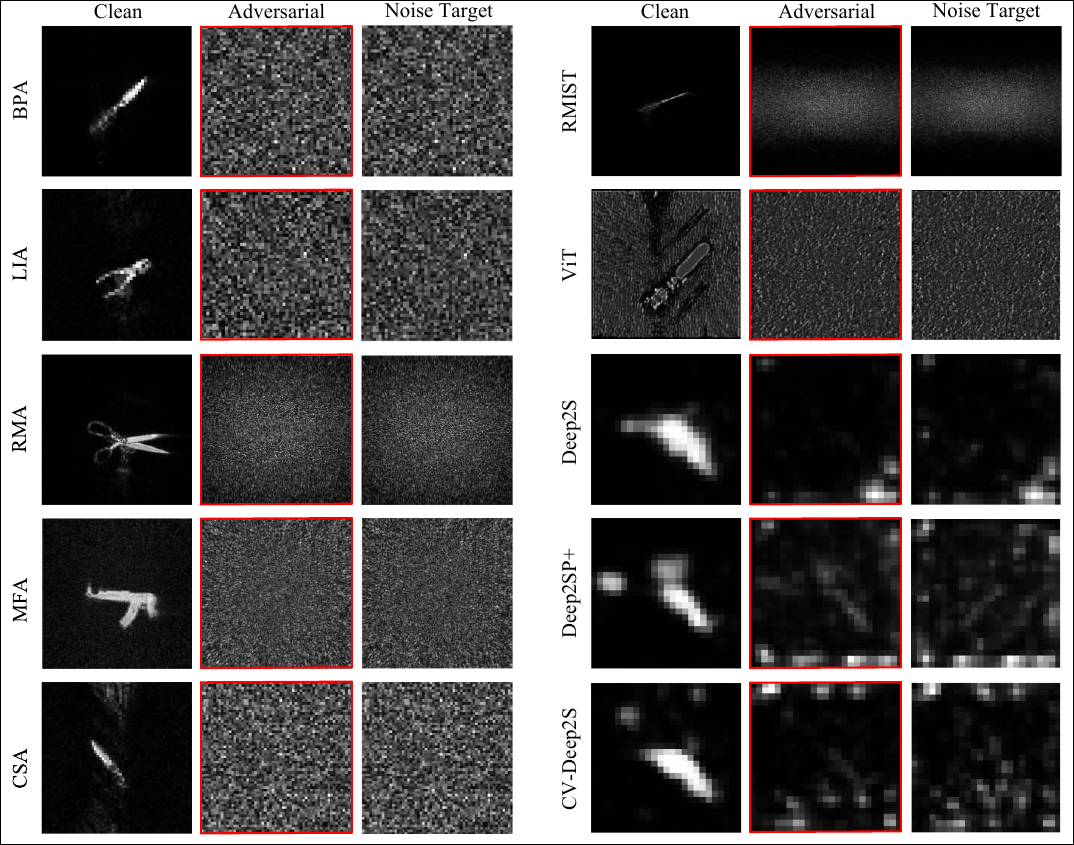}
\caption{Representative results of the Target Conceal Attack.  The clean reconstruction  $\hat{\boldsymbol{\alpha}}$ (left), the adversarial reconstruction $\hat{\boldsymbol{\alpha}}_{\mathrm{adv}}$(center), and the corresponding attack target $\hat{\boldsymbol{\alpha}}_{\mathrm{tgt}}$(right) are shown for each imaging algorithm. Across different imaging algorithms, the injected waveform effectively removes salient target features while maintaining visually plausible background structure.}

\label{fig:conceal_atk}
\end{figure*}


\begin{table*}[t]
\centering
\caption{Performance of the Target Swap Attack without power constraints. The adversarial images are highly similar to the target images, indicating successful swapping.}
\label{tab:target_swap_nopower}
\small
\setlength{\tabcolsep}{5pt}
\renewcommand{\arraystretch}{1.08}

\begin{tabular}{l|ccccc}
\toprule
Alg. & PSNR$_{AC}$ & SSIM$_{AC}$ & PSNR$_{AT}$ & SSIM$_{AT}$ & $P_a/P_s$ \\
\midrule
\footnotesize
BPA & $10.23 \pm 3.62$ & $ 0.2715 \pm 0.1421$ & $106.25 \pm 9.44$ & $ 1.0000 \pm 0.0000$ & $0.0960 \pm 0.0660$ \\
LIA & $10.65 \pm 3.62$ & $0.2824 \pm 0.1370$ & $34.05 \pm 3.33$ & $0.9706 \pm 0.0119$ & $0.0948 \pm 0.0671$ \\
RMA & $8.57 \pm 5.08$ & $0.3201 \pm 0.1484$ & $41.38 \pm 6.74$ & $0.9499 \pm 0.0795$ & $0.7802 \pm 0.6046$ \\
MFA & $9.86 \pm 5.13$ & $0.3189 \pm 0.0975$ & $32.22 \pm 4.72$ & $0.8763 \pm 0.0901$ & $0.8993 \pm 0.6911$ \\
CSA & $11.84 \pm 6.16$ & $0.2268 \pm 0.1310$ & $57.62 \pm 10.95$ & $0.9714 \pm 0.0395$ & $0.0634 \pm 0.0487$ \\
\hline
RMIST & $14.38 \pm 2.07$ & $0.3862 \pm 0.1100$ & $23.77 \pm 4.31$ & $0.6536 \pm 0.1532$ & $2.1740 \pm 0.9770$ \\
ViT & $18.55 \pm 1.00$ & $0.2231 \pm 0.0534$ & $28.05 \pm 1.55$ & $0.7979 \pm 0.0423$ & $1.3186 \pm 1.0774$ \\
Deep2S & $18.42 \pm 1.41$ & $0.6088 \pm 0.0323$ & $32.49 \pm 3.10$ & $0.9022 \pm 0.0450$ & $1.2362 \pm 0.4957$ \\
Deep2SP+ & $19.94 \pm 1.98$ & $0.6701 \pm 0.0445$ & $32.36 \pm 3.63$ & $0.8954 \pm 0.0444$ & $1.1779 \pm 0.4537$ \\
CV-Deep2S & $17.99 \pm 2.73$ & $0.6098 \pm 0.1141$ & $31.10 \pm 3.20$ & $0.8885 \pm 0.0428$ & $1.5233 \pm 0.8733$ \\
\bottomrule
\end{tabular}

\end{table*}

\begin{figure*}[t]
\centering
\includegraphics[width=\textwidth,
trim=1pt 1pt 1pt 1pt,clip]{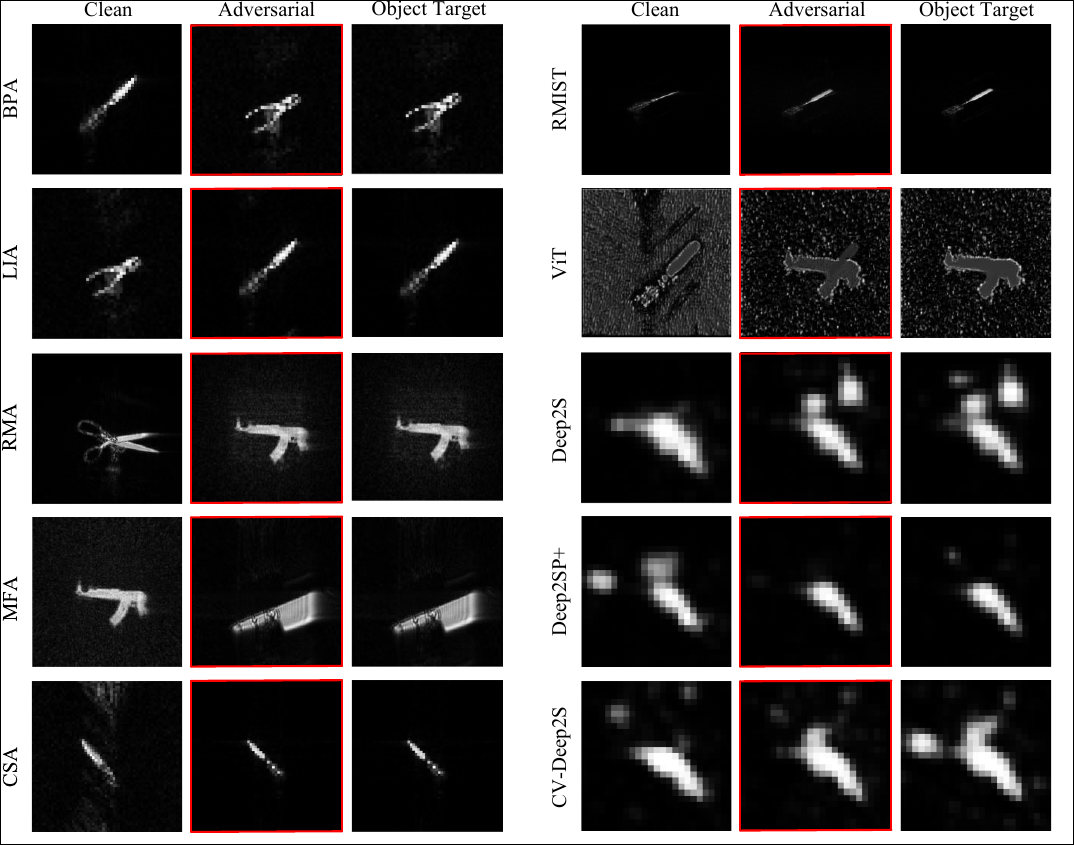}
\caption{Representative results of the Target Swap Attack. Each row corresponds to a different imaging algorithm, showing the clean reconstruction (left), the adversarial reconstruction (center), and the attacker's target image (right). The reconstructed adversarial images closely resemble the attacker-defined target images, demonstrating precise control over image content through waveform-domain manipulation.}

\label{fig:swap_atk}
\end{figure*}


\begin{table*}[t]
\centering
\caption{Performance of the Target Swap Attack under a power constraint ($P_a/P_s = 0.1$). The attack remains effective for classical algorithms, while learning-based methods exhibit reduced similarity to the target image, indicating better robustness.}
\label{tab:target_swap_power}
\small
\setlength{\tabcolsep}{5pt}
\renewcommand{\arraystretch}{1.08}
\begin{tabular}{l|ccccc}
\toprule
Alg. & PSNR$_{AC}$ & SSIM$_{AC}$ & PSNR$_{AT}$ & SSIM$_{AT}$ & $P_a/P_s$ \\
\midrule
\footnotesize
BPA & $10.96 \pm 3.84$ & $0.3275 \pm 0.1780$ & $83.97 \pm 40.43$ & $0.9552 \pm 0.0724$ & $0.10$ \\
LIA & $11.37 \pm 3.87$ & $0.3334 \pm 0.1706$ & $31.10 \pm 4.54$ & $0.9084 \pm 0.1220$ & $0.10$ \\
RMA & $12.73 \pm 5.27$ & $0.3777 \pm 0.1840$ & $19.66 \pm 6.23$ & $0.7871 \pm 0.1065$ & $0.10$ \\
MFA & $17.85 \pm 4.64$ & $0.4792 \pm 0.2124$ & $18.77 \pm 5.30$ & $0.7539 \pm 0.1223$  & $0.10$ \\
CSA & $12.20 \pm 6.45$ & $0.2707 \pm 0.1288 $ & $51.15 \pm 16.71$ & $0.9466 \pm 0.0969$ & $0.10$ \\
\hline
RMIST & $26.07 \pm 2.53$ & $0.8135 \pm 0.0678$ & $14.60 \pm 1.57$ & $0.5543 \pm 0.0760$ &  $0.10$\\
ViT & $19.56 \pm 0.91$ & $0.3481 \pm 0.1298$ & $24.95 \pm 2.59$ & $0.6010 \pm 0.1072$ & $0.10$ \\
Deep2S & $23.48 \pm 1.73$ & $0.8318 \pm 0.0841$ & $21.32 \pm 2.56$ & $0.7312 \pm 0.0475$ & $0.10$ \\
Deep2SP+ & $23.71 \pm 2.47$ & $0.8406 \pm 0.0815$ & $24.11 \pm 3.78$ & $0.7983 \pm 0.0438$ &   $0.10$\\
CV-Deep2S & $25.52 \pm 3.81$ & $0.8778 \pm 0.0745$ & $19.30 \pm 3.24$ & $0.7559 \pm 0.1195$ &  $0.10$ \\
\bottomrule
\end{tabular}
\end{table*}


\begin{table}[t]
\centering
\caption{Performance of the Randomization Attack using randomly generated weights $\mathbf{w}$. Waveform injection significantly degrades reconstruction quality, demonstrating that mmWave imaging is sensitive to even the random perturbations.}
\label{tab:target_conceal_random}
\small
\setlength{\tabcolsep}{5pt}
\renewcommand{\arraystretch}{1.08}
\begin{tabular}{l|ccc}
\toprule
Alg. & PSNR$_{AC}$ & SSIM$_{AC}$ & $P_a/P_s$ \\
\midrule
\footnotesize
BPA & $1.20 \pm 5.77$ & $0.0106 \pm 0.0102$ & $10$ \\ 
LIA & $1.51 \pm 5.54$ & $0.0110 \pm 0.0106$ & $10$ \\
RMA & $9.05 \pm 7.38$ & $0.0120 \pm 0.0101$ & $10$ \\
MFA & $1.70 \pm 6.96$ & $0.0054 \pm 0.0059$ & $10$ \\
CSA & $1.89 \pm 7.39$ & $0.0083 \pm 0.0096$ & $10$ \\
\hline
RMIST & $11.70 \pm 8.48$ & $0.1558 \pm 0.1384$ & $10$ \\
ViT  & $14.72 \pm 1.47$ & $0.0819 \pm 0.0349$ & $10$ \\
Deep2S & $14.10 \pm 1.32$ & $0.2301 \pm 0.0898$ & $10$ \\
Deep2SP+ & $15.31 \pm 1.49$ & $0.3210 \pm 0.0919$ & $10$ \\
CV-Deep2S & $13.47 \pm 2.28$ & $0.2373 \pm 0.0964$ & $10$ \\
\bottomrule
\end{tabular}
\end{table}

\begin{figure}[tbp]
\centerline{
\includegraphics[width=0.5\textwidth, trim=1pt 1pt 1pt 1pt, clip]{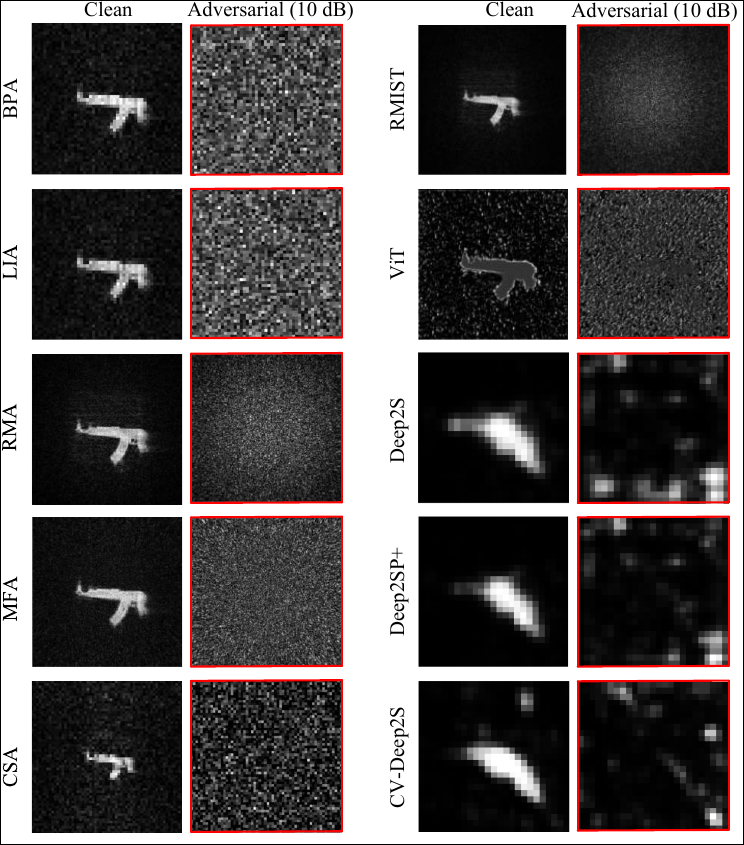}
}
\caption{Representative results of the Randomization Attack. Using randomly generated attack weights, the adversary induces blank reconstructed images similar to those of target conceal attacks.}
\label{fig:random_attack_success}
\end{figure}


\section{Conclusion} \label{sec:conclusion}

This paper presents a systematic study of the adversarial robustness of near-field mmWave imaging algorithms under waveform-domain physical attacks. A differential imaging attack framework is developed to optimize injected adversarial waveforms. To enable experimentally grounded evaluation, a dataset of real measured clean imaging waveforms and attack waveforms is constructed using a mmWave imaging testbed. The proposed framework is applied to 10 representative imaging algorithms.
The results demonstrate that near-field mmWave imaging algorithms are vulnerable to adversarial waveform injection, and that targeted manipulation of reconstructed image content is feasible. Learning-based imaging algorithms generally exhibit greater robustness than classical algorithms. 
In safety-critical applications such as airport passenger screening, such vulnerabilities could allow an adversary to conceal prohibited objects or inject false targets, posing significant risks to public security.

The vulnerability identified in this paper motivates the development of robustness-enhancing defense mechanisms for mmWave imaging systems. Potential defenses include waveform authentication or randomization to prevent attackers from accurately synchronizing injected signals, signal-level anomaly detection to identify inconsistencies in the received echoes across aperture looks, multi-sensor consistency checks to detect physically implausible reconstructions, and robust reconstruction algorithms that explicitly model adversarial perturbations, similar in spirit to adversarial training. Developing practical and provably robust defenses in the waveform domain remains an important direction for future work.

\appendix

\subsection{Derivation of mmWave Imaging Data Model (\ref{y=Hxmodel})}
\label{appendix:yHxmodel}

As described in Section \ref{subsec:forward} and Fig.~\ref{fig:sar_sat_img_2}, the imaging sensor moves over a planar aperture at $z=0$ with aperture locations $\mathbf r'=(x',y',0)$. The target reflectivity distribution lies on the plane $z=z_0$ and is denoted by $\sigma(\mathbf r)$ at location $\mathbf r=(x,y,z_0)$. 

At each aperture look $\ell$, the transmit antenna is located at $\mathbf r'_{\mathrm T,\ell}$ and the receive antenna at $\mathbf r'_{\mathrm R,\ell}$. Under the Born (single-scattering) approximation, the received signal at aperture look $\ell$ is given by \eqref{eq:rf_mix_new}, and the bistatic delay is
\begin{equation}
\tau_\ell(\mathbf r)
=
\frac{
|\mathbf r-\mathbf r'_{\mathrm T,\ell}|
+
|\mathbf r-\mathbf r'_{\mathrm R,\ell}|
}{c}
\stackrel{\triangle}{=}
\frac{R_{\mathrm T}(\mathbf{r})+R_{\mathrm R}(\mathbf{r})}{c},
\label{eq:prop_delay}
\end{equation}
where $c$ is the speed of light, and $R_T(\mathbf{r})$ and $R_R(\mathbf{r})$ denote the distances from the transmit and receive antennas to the target, respectively.

The conjugate of the received signal is mixed with the transmitted waveform: 
\begin{eqnarray}
\hat{s}_\ell(t)
&=&
{s}_\ell^*(t)\,p(t)  \nonumber \\
&=&
\iiint
\sigma(\mathbf r)\,
p^*(t-\tau_\ell(\mathbf{r}))p(t)
d\mathbf r
+
\tilde{\nu}_\ell(t).
\label{eq:dechirp_proc1}
\end{eqnarray}

Using the FMCW identity, we obtain
\begin{align}
p^*&(t-\tau_\ell(\mathbf{r}))p(t) \nonumber \\
&=
\exp\!\left\{
j2\pi f_0\tau_\ell(\mathbf{r})
+j2\pi K\tau_\ell(\mathbf{r}) t-j\pi K\tau_\ell^2(\mathbf{r})
\right\}.
\end{align}

Neglecting the term $-j\pi K\tau_\ell^2(\mathbf{r})$ \cite{wang2014application} and taking the Fourier transform with respect to $t$ produces peaks at the beat frequency $K\tau_\ell(\mathbf r)$. Selecting the frequency bin corresponding to the target plane $z=z_0$ yields
\begin{equation}
y_\ell
=
\iiint
\sigma(\mathbf r)\,
\frac{
\exp\!\big(j2\pi f_0\tau_\ell(\mathbf r)\big)
}{
R_{\mathrm T}(\mathbf r)\,R_{\mathrm R}(\mathbf r)
}
\,d\mathbf r
+
v_\ell,
\label{eq:range_compressed}
\end{equation}
where the amplitude term arises from spherical spreading.

Using \eqref{eq:prop_delay}, the phase term becomes $\exp\!\big(j2k(R_{\mathrm T}(\mathbf{r})+R_{\mathrm R}(\mathbf{r}))\big)$. 
Define the equivalent one-way propagation distance as 
\begin{equation}
R_0(\mathbf r'_\ell,\mathbf r)
=
\frac{
R_{\mathrm T}(\mathbf r)
+
R_{\mathrm R}(\mathbf r)
}{2}.
\end{equation}

Then \eqref{eq:range_compressed} can be rewritten as
\begin{equation}
y_\ell
=
\iiint
\sigma(\mathbf r)\,
\frac{
\exp\!\big(j2kR_0(\mathbf r'_\ell,\mathbf r)\big)
}{
R_0^2(\mathbf r'_\ell,\mathbf r)
}
\,d\mathbf r
+
v_\ell.
\end{equation}

Since the reflectivity is confined to the plane $z=z_0$, we write $\sigma(\mathbf r)=\sigma(x,y)\,\delta(z-z_0)$, which yields
\begin{equation}
y_\ell
=
\iint
\sigma(x,y)\,
\frac{
\exp\!\big(j2kR_0(\mathbf r'_\ell,\mathbf r)\big)
}{
R_0^2(\mathbf r'_\ell,\mathbf r)
}
\,dx\,dy
+
v_\ell.
\label{eq:continuous_model}
\end{equation}

Discretizing the imaging region into $N$ pixels centered at $\{\mathbf r_n\}_{n=0}^{N-1}$ and denoting the reflectivity by $\alpha_n$, we approximate the integral in \eqref{eq:continuous_model} by a Riemann sum:
\begin{equation}
y_\ell
\approx
\sum_{n=0}^{N-1}
\alpha_n\,
\frac{
\exp\!\big(j2kR_0(\mathbf r'_\ell,\mathbf r_n)\big)
}{
R_0^2(\mathbf r'_\ell,\mathbf r_n)
}
+
v_\ell.
\label{eq:dechirp_proc2}
\end{equation}

Stacking all $L$ aperture-look samples $y_\ell$, $\ell=0, \cdots, L-1$, yields the linear system \eqref{y=Hxmodel}. The matrix $\mathbf H$ represents the exact near-field propagation operator under the Born approximation.

\subsection{Proof of Theorem \ref{theorem1} of Section \ref{subsec:rf_injection} }
\label{app:delay_equiv}

Accurate frequency and timing synchronization are required for the attack waveform $p_\ell^{\mathrm{atk}}(t)$ to be correctly received and processed by the imaging sensor \cite{miura2019low, ordean2022millimeter}. First, given the clean sensing waveform $p(t)$ in \eqref{eq:clean_tx_wave}, the attack waveform must share the same start frequency $f_0$ and chirp slope $K$.

Next, consider the propagation delay $\tau_{\ell}(\mathbf{r})$ of the clean echo defined in \eqref{eq:prop_delay}. The propagation delay of the attack signal is $\tau_\ell^{\mathrm{atk}}(\mathbf{r}_A) = |\mathbf{r}_A - \mathbf{r}'_{R,\ell}|/c$, where $\mathbf{r}_A$ denotes the attacker location. Since $\tau_\ell(\mathbf{r}) \neq \tau_\ell^{\mathrm{atk}}(\mathbf{r}_A)$, the attack waveform must be time-shifted to align with the desired frequency bin corresponding to the target plane $z=z_0$. This is achieved by transmitting $p(t-\beta_\ell)$, where
\begin{equation}
    \beta_\ell = \tau_{\ell}(\mathbf{r}) - \tau_\ell^{\mathrm{atk}}(\mathbf{r}_A). 
\end{equation}

In addition, a complex weight $c_\ell$ is applied to control the amplitude and phase, resulting in the delay-and-weight model
\begin{equation}
p^{\mathrm{atk}}_{\ell}(t)=c_{\ell}\,p(t-\beta_{\ell}). 
\label{eq:atk_delay_weight}
\end{equation}

A time-shifted FMCW chirp can be equivalently expressed as the original chirp multiplied by a complex constant and a single-tone modulation:
\begin{align}
p(t-\beta_{\ell})
&=\exp\!\left\{j2\pi\!\left(f_0(t-\beta_{\ell})+\tfrac{1}{2}K(t-\beta_{\ell})^2\right)\right\} \nonumber\\
&=p(t)\,\exp\!\left\{j2\pi\!\left(-\beta_{\ell}(f_0+Kt)+\tfrac{1}{2}K\beta_{\ell}^2\right)\right\} \nonumber\\
&=p(t)\,e^{j\varphi_\ell}\,e^{-j2\pi \delta_\ell t},
\label{eq:chirp_shift_equiv}
\end{align}
where
\begin{equation}
\delta_\ell \triangleq K\beta_\ell, \;\;\;\;
\varphi_\ell \triangleq -2\pi f_0\beta_\ell + \pi K\beta_\ell^2.
\end{equation}

Letting $w_\ell = c_\ell e^{j\varphi_\ell}$ yields \eqref{eq:atk_tx_new}. Following the same processing steps as in \eqref{eq:dechirp_proc1}--\eqref{eq:dechirp_proc2}, we obtain \eqref{eq:atk_sig_after_proc}.

The delay $\beta_\ell$ can be estimated from the timing and geometric configuration of the imaging sensor, attacker, and target. Precomputing $\beta_\ell$ and $\delta_\ell$ simplifies the optimization of $w_\ell$. Otherwise, incorporating the FMCW dechirping process into the optimization would significantly increase complexity and may lead to numerical instability due to discrete frequency bin selection.

\bibliographystyle{IEEEtran}
\bibliography{mybibs, mybib}

\end{document}